\documentclass{article}
\usepackage{amsfonts}
\usepackage{amsmath}
\usepackage{amssymb}
\usepackage{amsthm}
\usepackage{fullpage}
\usepackage[hidelinks]{hyperref}
\usepackage{cleveref}
\usepackage{mathtools}
\usepackage{cite}
\newtheorem{theorem}{Theorem}[section]
\newtheorem{claim}[theorem]{Claim}

\newtheorem{lemma}[theorem]{Lemma}

\theoremstyle{definition}

\newcommand{\eps}{\varepsilon}

\newcommand{\E}{\mathbb E}
\newcommand{\F}{\mathbb F}
\newcommand{\one}[1]{\mathbf 1 \left[#1\right]}

\title{Approximating Independent Sets Using Additive Combinatorics}
\author{Shachar Lovett\thanks{Supported by Simons Investigator Award
\#929894 and NSF award CCF-2425349.}\\
University of California, San Diego\\
\texttt{slovett@ucsd.edu}}

\begin{document}

\maketitle

\begin{abstract}
We give a deterministic quasi-polynomial time algorithm for approximating
the largest independent set in a graph. We show that if an $N$-vertex graph contains an independent set of size $2^{-t} N$, then we can find an independent set of size $2^t$ in quasi-polynomial time, for all $t$ up to a constant factor of $(\log N)^{1/3}$. This gives an approximation factor of $N / \exp((\log N)^{1/3})$, which has the same shape as known hardness results for quasi-polynomial time algorithms, up to a different exponent of the logarithm.

The algorithm builds on a structural
principle: a graph containing a large independent set must also contain a
large sparse structured set, and such a set can be found efficiently. A greedy
algorithm then extracts an independent set from it. Our proof draws on tools in additive combinatorics, in particular
spreadness and sifting, to obtain the combinatorial structural results. We then convert these to efficient algorithms using pseudo-marginals obtained by semidefinite programming.
\end{abstract}

\section{Introduction}

Finding a maximum independent set is a classical NP-hard problem~\cite{karp1972reducibility}.
Throughout, $G=(V,E)$ is a simple undirected graph on $N$ vertices, and
$\alpha(G)$ denotes its maximum independent-set size. An $r$-approximation
algorithm returns an independent set of size at least $\alpha(G)/r$.

\paragraph{Polynomial time approximation.}
Feige~\cite{feige2004approximating} gives the best known polynomial time approximation
ratio,
\[
  O\!\left(\frac{N(\log\log N)^2}{(\log N)^3}\right).
\]
On the hardness side, H\aa stad~\cite{hastad1999clique} proved that, for every constant
$\eps>0$, an $N^{1-\eps}$-approximation is impossible in polynomial time
unless $\mathrm{NP}=\mathrm{ZPP}$. Zuckerman~\cite{zuckerman2007linear} derandomized
this result, establishing NP-hardness for the same factor. So, the gap between known algorithms and hardness results is large: polylogarithmic versus small polynomial savings over the trivial factor of $N$.

\paragraph{Quasi-polynomial time approximation.}
We now switch to consider algorithms that are allowed to run in quasi-polynomial time, namely  $N^{O(\log N)^a}$ time for any constant $a>0$. Within this time limit, exhaustive search gives a an approximation ratio of
$O(N/(\log N)^a)$, by enumerating all vertex sets of size
at most $(\log N)^a$.
The enumeration either finds an optimum independent set or finds one of
size $(\log N)^a$.
On the hardness side, building on Khot~\cite{khot2001improved},
Zuckerman~\cite{zuckerman2007linear} ruled out quasi-polynomial time approximation
within $N/2^{(\log N)^{1-\gamma}}$ for some constant $\gamma >0$, unless
NP has deterministic quasi-polynomial time algorithms.
Under the stronger assumption that NP has no bounded-error randomized
quasi-polynomial time algorithms, Khot and Ponnuswami~\cite{khot2006better} ruled out
polynomial time approximation within $N/2^{(\log N)^{3/4+\eps}}$ for every
constant $0<\eps<1/4$.

In this paper, we design a new quasi-polynomial time algorithm that achieves the approximation ratio
\[
  N/2^{\Omega(\log N)^{1/3}}.
\]
This has the same form as the hardness bound
of~\cite{khot2006better}, though the exponents differ: $1/3$ in our algorithm and
$3/4+\eps$ in the hardness bound. Formally, we prove the following theorem.

\begin{theorem}[Finding a large independent set]\label{thm:ind}
There is an absolute constant $c>0$ such that the following holds.
Let $t$ be an integer with $1\le t\le c(\log N)^{1/3}$. If $\alpha(G) \ge 2^{-t} N$, then a deterministic algorithm finds an
independent set in $G$ of size at least $2^t$ in time $N^{O(t)}$.
\end{theorem}

To obtain the approximation guarantee, take
$t= c(\log N)^{1/3}$ and run the algorithm, returning a
singleton if it finds no independent set. If $\alpha(G)\ge N/2^t$,
the output has size at least $2^t$, giving an approximation factor of at
most $N/2^t$. If $\alpha(G)<N/2^t$, a singleton already gives this factor.
The running time is $N^{O((\log N)^{1/3})}$ which is quasi-polynomial in $N$.

\paragraph{Large sparse structured sets.}

The algorithm in \Cref{thm:ind} is based on two steps. First, we find a large sparse \emph{structured} set in the graph, and then extract an independent set from it by a simple greedy algorithm. The first step is the main innovation of this paper. It is a manifestation of the following guiding principle: 

\begin{center}
\begin{minipage}{0.9\linewidth}
\centering\itshape
A large exact solution guarantees a large structured approximate
solution,\\ with the structure captured by a small, efficiently
searchable family.
\end{minipage}
\end{center}
In our case, an exact solution is a large independent set, and an approximate solution is a set with few edges. The specific structure we utilize will become apparent soon. 

We now move to the formal statement of our main technical theorem.
For a nonempty set of vertices $U\subseteq V$, define its edge
density by
\[
  d(U)\coloneq\Pr_{u,v\in U}[uv\in E]
      =\frac{2|E(G[U])|}{|U|^2}.
\]
Here $u,v$ are chosen independently and uniformly from $U$.

\begin{theorem}[Finding a large sparse set]\label{thm:sparse}
Let $t,h\ge1$ be integers. If $\alpha(G)\ge2^{-t}N$, then a deterministic
algorithm running in time $N^{O(t)}$ finds a nonempty set $U\subseteq V$
satisfying
\[
  |U|\ge2^{-O(t^2h)}N,\qquad d(U)\le2^{-h}.
\]
\end{theorem}

Before describing the proof approach for \Cref{thm:sparse}, we first show that
\Cref{thm:ind} follows from it by simple greedy pruning.

\begin{proof}[Proof of \Cref{thm:ind} given \Cref{thm:sparse}]
Apply \Cref{thm:sparse} with $h=t+1$. For a sufficiently small choice of $c$, the resulting set $U$ satisfies
\[
  |U|\ge N2^{-O(t^3)}\ge2^{t+1},\qquad d(U)\le2^{-(t+1)}.
\]
We work in the induced subgraph $G[U]$.
Consider the greedy algorithm that iteratively selects a vertex of minimum degree, adds it to the independent set, and deletes it together with its neighbors. This algorithm is well-known to attain Tur\'an's bound
(see, e.g.,~\cite{halldorsson1997greed}), finding an independent set in $G[U]$ of size at least
\[
  \frac{|U|}{1+d(U)|U|}
  =\frac{1}{|U|^{-1}+d(U)}\ge2^t.
\]
\end{proof}

\subsection{Proof overview}
We first illustrate our guiding principle for Cayley graphs over $\F_2^n$, where
existing additive-combinatorial results supply the structured sparse sets.
Our proof for general graphs uses some of the same ideas directly, in particular
spreadness and sifting, developed by Kelley and Meka~\cite{kelley2023strong}
in the arithmetic setting and by Kelley, Lovett, and Meka~\cite{kelley2024explicit}
in the combinatorial setting.

\paragraph{Cayley graphs.}
Consider a Cayley graph $G$ on $\mathbb F_2^n$, with $N=2^n$ vertices and
generating set $S\subseteq\mathbb F_2^n\setminus\{0\}$: vertices $x,y$ are
adjacent when $x+y\in S$. If $A$ is an independent set in $G$, then
\[
  (A+A)\cap S=\emptyset.
\]
We will use the additive structure in $A+A$ to produce a structured sparse set in $G$.

Assume $|A|\ge2^{-t}N$ for an integer $t\ge1$. The quasi-polynomial Bogolyubov lemma, following the
Kelley-Meka framework~\cite{kelley2023strong} in the form proved by Hunter and
Pohoata~\cite{hunter2025off}, gives, for every integer $h\ge1$, a linear subspace
$U\le\mathbb F_2^n$ such that
\[
  \operatorname{codim}(U)=O(t^5 h^2),\qquad
  |U\setminus(A+A)|\le2^{-h}|U|.
\]
Since $S$ avoids $A+A$, at most a $2^{-h}$ fraction of $U$ lies in $S$.
For independent uniform $x,y\in U$, the sum $x+y$ is uniform in $U$.
Consequently,
\[
  |U|\ge N2^{-O(t^5 h^2)},\qquad
  d(U)=\frac{|S\cap U|}{|U|}\le2^{-h}.
\]
This is the structural principle in a concrete form: $U$ is large and
sparse, and its structure is that of a linear subspace. A subspace of
codimension $k$ is the common zero set of $k$ linear forms, for which there are at most $N^k$ choices. Enumerating subspaces of codimension at most
$O(t^5h^2)$ and checking their edge densities therefore finds a suitable
$U$ in deterministic time $N^{O(t^5h^2)}$.

\paragraph{General graphs.}
We begin by assuming access to a distribution on independent sets
whose expected size is large. This lets us first develop the
combinatorial argument. We then show that access to the distribution
can be replaced by access to its marginals, namely the probabilities
that small sets of vertices are contained in the random independent
set. Finally, we show that pseudo-marginals computed by an SDP suffice
to implement the argument efficiently.

The first step is to increase inclusion probabilities by conditioning.
If a set of vertices contains many edges and each vertex has a
reasonably large inclusion probability, \Cref{lemma:increment} shows
that conditioning on one vertex belonging to the independent set
increases the inclusion probabilities of many other vertices by a
constant factor.

We then iterate this step, retaining the vertices whose inclusion
probabilities increase. Since probabilities cannot exceed one, the
process must stop with a sparse set. In \Cref{lem:iterate}, we track
the number of retained vertices to show that this sparse set is large.

To make the argument algorithmic, we show that it uses only small
marginals and remains valid for pseudo-marginals obtained from a
semidefinite relaxation. Each candidate sparse set is determined by
a small conditioning set. Enumerating these conditioning sets and
checking the resulting candidates gives an efficient algorithm.

\subsection{Related work}

Our work builds on a long line of research in additive combinatorics
and approximation algorithms. Its main contribution is to use
additive-combinatorial techniques to obtain strong quantitative
guarantees from semidefinite relaxations of the independent set problem.

\paragraph{Finding almost-cliques and sparse sets.}
Algorithms for finding almost-cliques
\cite{feige1997densest,barman2018approximating} already yield sparse
sets when large independent sets exist. These algorithms preserve
the set size; we allow it to shrink while achieving much smaller
edge density within quasi-polynomial time. Related hardness results appear
in~\cite{braverman2017eth}.

\paragraph{From distributions to SDP rounding.}
Barak, Kelner, and Steurer~\cite{barak2014rounding} develop a framework
for converting arguments about distributions into SDP rounding
algorithms. We follow this approach, using only small marginals.
High entrywise powers of the correlation matrix require no
additional SDP levels.

\paragraph{Correlation rounding and conditioning.}
Conditioning is central to correlation
rounding~\cite{barak2011rounding,raghavendra2012approximating},
where progress is measured through variance or entropy.
Our iteration increases all retained inclusion marginals by a
constant factor, bounding the number of conditioning steps.

\paragraph{Odd correlation moments and spectral methods.}
Jin, Milojevi\'c, Tomon, and Zhang~\cite{jin2025small} use entrywise
cubes and the Schur product theorem to obtain dense subgraphs under
spectral assumptions. We apply a related odd-moment argument to
a normalized covariance matrix, obtaining positive correlations
that yield a marginal increment.

\paragraph{Hierarchies for independent set and densest subgraph.}
Local distributions and small conditioning sets also appear in
independent-set algorithms
\cite{guruswami2013rounding,bansal2018lovasz}
and densest-subgraph algorithms~\cite{bhaskara2010detecting}.
Our argument uses these tools to find sparse sets, with guarantees
depending only on the size of the promised independent set.

\paragraph{Containers and additive-combinatorial sifting.}
Container theorems~\cite{balogh2015independent,saxton2015hypergraph}
provide another precedent for sparse sets encoded by small vertex
subsets. Our density-increment argument draws on
Kelley--Meka and Kelley--Lovett--Meka
\cite{kelley2023strong,kelley2024explicit};
Bloom--Sisask~\cite{bloom2023kelley} explain the connection between
sifting and dependent random choice.

\subsection{Open problems}

The most immediate open problem is to narrow the gap between algorithms
and hardness results for independent set approximation. A broader and,
in our view, more interesting question is where else techniques from
additive combinatorics can yield stronger quantitative guarantees for
approximation algorithms. In particular, can we find further
realizations of our guiding principle, where a large exact solution
guarantees a large structured approximate solution belonging to a small,
efficiently searchable family?

\paragraph{Paper organization.} We prove \Cref{thm:sparse} in \Cref{sec:proof} and describe our AI methodology in \Cref{sec:ai}.

\section{Finding large sparse structured sets, efficiently}
\label{sec:proof}

We prove \Cref{thm:sparse} in this section. We first work with a
distribution $\mu$ over independent sets whose expected size is large.
We then show that the argument uses only small marginals and remains
valid for pseudo-marginals, which can be computed efficiently by
semidefinite programming.

\paragraph{Notation and definitions.}
Let $G=(V,E)$ be a graph on $N$ vertices. For a distribution $\mu$ on
the independent sets of $G$ and a set $T\subseteq V$, write
\[
  \mu(T)\coloneq\Pr_{I\sim\mu}[T\subseteq I].
\]
Thus $\mu(T)$ denotes the probability that the random independent set
\emph{contains} $T$. We abbreviate $\mu(\{v\})$ as $\mu(v)$ and
$\mu(\{u,v\})$ as $\mu(u,v)$. In particular,
\begin{equation}\label{eq:marginal-bounds}
  0\le\mu(v)\le1,\qquad
  0\le\mu(u,v)\le\min\{\mu(u),\mu(v)\},\qquad
  \mu(u,v)=0\quad\text{if }uv\in E.
\end{equation}
For $S\subseteq V$ with $\mu(S)>0$, let $\mu_S$ be the distribution of
$I\sim\mu$ \emph{conditioned} on $S\subseteq I$, where
\[
  \mu_S(T)=\frac{\mu(S\cup T)}{\mu(S)}.
\]
We write $\mu_v$ for $\mu_{\{v\}}$. Successive conditionings satisfy
$(\mu_S)_T=\mu_{S\cup T}$ whenever $\mu(S\cup T)>0$.

\subsection{Increasing marginals by conditioning}

The first step is inspired by the sifting and spectral-positivity arguments
of Kelley and Meka~\cite{kelley2023strong} and their extension to general bipartite graphs
through grid norms by Kelley, Lovett, and Meka~\cite{kelley2024explicit}. The main idea is to use high moments of correlations in a distribution on
independent sets to obtain a density increment by conditioning.

\begin{lemma}[Increased marginals in dense subgraphs]\label{lemma:increment}
Let $\mu$ be a distribution on independent sets of $G$, and let
$U\subseteq V$. Let $\delta,q \in (0,1)$
and suppose $d(U)\ge\delta$ and
$\mu(u)\ge q$ for every vertex $u\in U$.
Let $\ell$ be a positive odd integer satisfying $2^{-\ell}\le\delta/2$.
Then some $v\in U$ satisfies
\[
  \left|\left\{u\in U:\mu_v(u)\ge\tfrac32\mu(u)\right\}\right|
  \ge\frac\delta2q^\ell|U|.
\]
\end{lemma}

\begin{proof}
Define a matrix $M$ indexed by $u, v \in U$ as
\[
M(u,v) = \frac{\mu(u,v)}{\mu(u)\mu(v)}-1.
\]
We claim that $M$ is positive semidefinite. To see this,
consider the matrix $R$ indexed by $u \in U$ and an independent set $I$, given by $R(u,I)=\frac{\one{u\in I}}{\mu(u)}-1$, and note that $M(u,v)=\E_{I\sim\mu}R(u,I)R(v,I)$.
The marginal bounds on $\mu$ in \Cref{eq:marginal-bounds} and our assumptions give
\[
  M(u,v)=-1\quad\text{if }uv\in E,\qquad
  M(u,v)\le q^{-1}\quad\text{for all }u,v\in U.
\]
Write $M^{\circ\ell}=(M(u,v)^\ell)_{u,v\in U}$ for the entrywise
$\ell$-th power of $M$. Since $M$ is positive semidefinite, by the Schur product theorem~\cite{schur1911bemerkungen},
$M^{\circ\ell}$ is also positive semidefinite, and so
\[
  \E_{u,v\in U}M(u,v)^\ell
  =\frac{\boldsymbol{1}^{\mathsf T}M^{\circ\ell}\boldsymbol{1}}{|U|^2}
  \ge0,
\]
where $\boldsymbol{1}$ is the all-ones vector. This nonnegativity can also
be seen directly:
\[
  \E_{u,v\in U} \left[ M(u,v)^\ell \right]
  =\E_{I_1,\ldots,I_\ell\sim\mu}
    \left(\E_{u\in U}\prod_{j=1}^{\ell}R(u,I_j)\right)^2
  \ge0.
\]
We next consider the contribution of each term $M(u,v)$ to the $\ell$-th moment.
Since $\ell$ is odd, each edge contributes $-1$ and hence
edges contribute $-d(U)$ in total.
Other negative terms contribute a negative amount. Terms with
$0\le M(u,v)<1/2$ contribute at most $2^{-\ell}$ in total, while those
with $M(u,v)\ge1/2$ contribute at most $\rho q^{-\ell}$, where  $\rho=\Pr_{u,v\in U}[M(u,v)\ge1/2]$. Thus
\[
  0\le\E_{u,v\in U}\left[ M(u,v)^\ell \right]
    \le-d(U)+2^{-\ell}+\rho q^{-\ell}
    \le-\frac\delta2+\rho q^{-\ell}.
\]
This implies that $\rho\ge(\delta/2)q^\ell$. By averaging, some $v\in U$ has at
least $\rho|U|$ vertices $u\in U$ with $M(u,v)\ge1/2$. For each such
vertex $u$,
\[
  \mu_v(u)=\frac{\mu(u,v)}{\mu(v)}\ge\tfrac32\mu(u),
\]
which proves the claim.
\end{proof}

\subsection{Finding a large sparse set by repeated conditioning}

We now iterate \Cref{lemma:increment}. Each conditioning multiplies the marginal
threshold by $3/2$. This process must stop because inclusion marginals are bounded by $1$.

\begin{lemma}[A large sparse set after conditioning]\label{lem:iterate}
Let $t,h\ge1$ be integers, and assume $N\ge2^{3(t+2)^2(h+2)}$.
Let $\mu$ be a distribution on independent sets of $G$ such that
$\E_{I\sim\mu}|I|\ge2^{-t}N$.
For each $S\subseteq V$ with $\mu(S)>0$, define
\[
  U_S=\left\{u\in V\setminus S:
    \mu_S(u)\ge(3/2)^{|S|}2^{-(t+1)}\right\}.
\]
There is such an $S$ with $|S|\le2(t+1)$ for which
\[
  |U_S|\ge2^{-3(t+2)^2(h+2)}N,\qquad d(U_S)\le2^{-h}.
\]
\end{lemma}

\begin{proof}
Start with $S=\emptyset$. Since
\[
  2^{-t}N\le\E_{I\sim\mu}|I|
    =\sum_{u\in V}\mu(u)\le|U_\emptyset|+2^{-(t+1)}N,
\]
we have $|U_\emptyset|\ge2^{-(t+1)}N$.
Let $K=2^{(t+2)(h+2)}$. We will maintain the invariant
$|U_S|\ge2^{-(t+1)} K^{-|S|} N$ throughout the process.

Suppose $\mu(S)>0$ and $d(U_S)>2^{-h}$.
Apply \Cref{lemma:increment} to $\mu_S$ and $U_S$, using the marginal lower
bound $\mu_S(u) \ge (3/2)^{|S|} 2^{-(t+1)} \ge 2^{-(t+1)}$ for all $u \in U_S$, and choosing $\ell\in\{h+1,h+2\}$ to be odd.
We obtain a vertex $v\in U_S$ for which at least
\[
  2^{-(h+1)-(t+1)\ell}|U_S|\ge\frac{2|U_S|}{K}
\]
vertices $u\in U_S$ satisfy
$\mu_{S\cup\{v\}}(u)\ge\tfrac32\mu_S(u)$.
After removing $v$, these vertices lie in $U_{S\cup\{v\}}$. We note that the lower bound on $|U_S|$ and our assumption that $N$ is large enough imply $|U_S| \ge K$.
Hence
\[
  |U_{S\cup\{v\}}|\ge\frac{2|U_S|}{K}-1\ge\frac{|U_S|}{K} \ge 2^{-(t+1)}  K^{-(|S|+1)} N,
\]
which maintains the invariant.
Since $S\cap U_S=\emptyset$, this adds one vertex to $S$; moreover,
$\mu(S\cup\{v\})=\mu(S)\mu_S(v)>0$.

Repeat this process as long as $d(U_S)>2^{-h}$. The process must stop before $|S|=2(t+1)$:
at that point the marginal threshold would be
$(3/2)^{2(t+1)}2^{-(t+1)}=(9/8)^{t+1}>1$, forcing $U_S=\emptyset$
and contradicting the size bound $|U_S|\ge K$ that we showed. When the process terminates at some set $S$, we have $d(U_S)\le2^{-h}$ and
\[
  |U_S|\ge 2^{-(t+1)} K^{-|S|} N  \ge 2^{-(t+1)}K^{-2(t+1)} N
    \ge2^{-3(t+2)^2(h+2)}N.
\]
\end{proof}

\subsection{Finding large sparse sets using pseudo-marginals}

We now turn our combinatorial observations into efficient algorithms. The first step is to note that in order to apply \Cref{lem:iterate} we do not need full access to $\mu$; it suffices to get access to small marginals of $\mu$. In fact, these do not have to be marginals of a real distribution; a pseudo-distribution is sufficient. The marginals of such a pseudo-distribution can be found by a semidefinite program.

\paragraph{Marginals suffice.}
Suppose we have access to the marginals $\mu(T)$ for all $|T|\le2t+3$, where
$\mu$ satisfies the assumptions of \Cref{lem:iterate}.
Enumerate all $S\subseteq V$ with $|S|\le2(t+1)$ and $\mu(S)>0$.
For each $S$, compute
\[
  \mu_S(u)=\frac{\mu(S\cup\{u\})}{\mu(S)},\qquad
  U_S=\left\{u\in V\setminus S:
    \mu_S(u)\ge(3/2)^{|S|}2^{-(t+1)}\right\}.
\]
Check the size and edge density of $U_S$ directly in $G$ and return a set
satisfying the conclusion of the lemma. The lemma guarantees that one
exists. There are $N^{O(t)}$ candidates, and each can be checked in polynomial time.
Thus small marginals suffice; we do not need to sample full independent sets from $\mu$.

\paragraph{Computing pseudo-marginals by SDP.}
A level-$r$ pseudo-distribution on independent sets is represented by numbers
$\widetilde\mu(T)$ for $T\subseteq V$, $|T|\le2r$, satisfying
\begin{equation}\label{eq:pseudo-marginals}
  \widetilde\mu(\emptyset)=1,\qquad
  \widetilde\mu(T)=0\ \text{if $T$ contains an edge},\qquad
  \bigl(\widetilde\mu(A\cup B)\bigr)_{|A|,|B|\le r}\succeq0.
\end{equation}
These numbers need not be the marginals of an actual distribution.
The last constraint requires the matrix indexed by sets $A,B$ of
size at most $r$, with entry $\widetilde\mu(A\cup B)$, to be positive
semidefinite. This matrix is often called the \emph{moment matrix}.
Maximizing $\sum_{u\in V}\widetilde\mu(u)$ subject to
\eqref{eq:pseudo-marginals} is the standard Lasserre SDP relaxation for
independent set~\cite{laurent2003comparison}.
We use the standard algorithmic guarantee: if $\alpha(G)\ge2^{-t}N$,
then in time $N^{O(r)}\operatorname{poly}(t)$ one can compute rational
feasible pseudo-marginals with
\begin{equation}\label{eq:pseudo-size}
  \sum_{u\in V}\widetilde\mu(u)\ge2^{-(t+1)}N.
\end{equation}
The factor-two slack allows approximate optimization of the objective;
the marginal constraints remain exact. For the bit-complexity guarantee
for this relaxation, see~\cite{raghavendra2017bit,gribling2023note}.

We verify that these pseudo-marginals can replace $\mu$ in the combinatorial
argument.

\begin{claim}\label{claim:pseudo-marginal-bounds}
At level $r\ge2$, the pseudo-marginals satisfy the bounds in
\Cref{eq:marginal-bounds}, with $\widetilde\mu$ in place of $\mu$.
\end{claim}

\begin{proof}
The principal submatrix on $\emptyset,\{u\}$ gives
$0\le\widetilde\mu(u)\le1$, the diagonal entry indexed by $\{u,v\}$ gives nonnegativity,
and the quadratic form on the difference of the indices $\{u\}$ and
$\{u,v\}$ gives $\widetilde\mu(u)-\widetilde\mu(u,v)\ge0$.
The edge marginals vanish by \Cref{eq:pseudo-marginals}.
\end{proof}

Conditioning uses the same notation and formula as before:
\[
  \widetilde\mu_S(T)=
    \frac{\widetilde\mu(S\cup T)}{\widetilde\mu(S)}
    \qquad\text{when }|S|\le r\text{ and }\widetilde\mu(S)>0.
\]
This gives a level-$(r-|S|)$ pseudo-distribution. Its moment matrix has entries
\[
  \widetilde\mu_S(A\cup B)
  =\frac{\widetilde\mu(A \cup B \cup S)}
         {\widetilde\mu(S)},\qquad |A|,|B|\le r-|S|,
\]
so it is obtained by selecting (possibly repeating) rows and columns $A \cup S, B \cup S$ of
the original moment matrix and dividing by $\widetilde\mu(S)>0$.
It is therefore also positive semidefinite. Normalization and the edge
constraints are also preserved.
Thus each conditioning consumes one level of the pseudo-distribution.

It remains to justify the nonnegative moment used in
\Cref{lemma:increment}. On any set $U$ with positive marginals, define
\[
  M(u,v)=\frac{\widetilde\mu(u,v)}
                    {\widetilde\mu(u)\widetilde\mu(v)}-1.
\]
\begin{claim}\label{claim:covariance-psd}
The matrix $M$ is positive semidefinite.
\end{claim}

\begin{proof}
Let $(a_u)_{u\in U}$ be arbitrary real coefficients. Define coefficients
$b_A$, indexed by sets $A$ of size at most $r$, by
$b_{\{u\}}=a_u/\widetilde\mu(u)$ for $u\in U$,
$b_\emptyset=-\sum_{u\in U}a_u$, and $b_A=0$ elsewhere.
Positive semidefiniteness of the moment matrix gives
\[
\begin{aligned}
  0&\le\sum_{\substack{A,B\subseteq V\\|A|,|B|\le r}}
       b_A b_B\,\widetilde\mu(A\cup B)\\
   &=\sum_{u,v\in U}a_ua_v
    \frac{\widetilde\mu(u,v)}{\widetilde\mu(u)\widetilde\mu(v)}
  -2\left(\sum_{u\in U}a_u\right)^2
  +\left(\sum_{u\in U}a_u\right)^2\\
   &=\sum_{u,v\in U}a_ua_vM(u,v).
\end{aligned}
\]
Since this holds for every choice of the coefficients, $M\succeq0$.
\end{proof}

By the Schur product theorem~\cite{schur1911bemerkungen}, every entrywise positive integer power
$M^{\circ\ell}=(M(u,v)^\ell)_{u,v\in U}$ is positive semidefinite. Hence
\[
  \E_{u,v\in U}M(u,v)^\ell
    =\frac{\boldsymbol{1}^{\mathsf T}M^{\circ\ell}\boldsymbol{1}}
           {|U|^2}\ge0,
\]
where $\boldsymbol{1}$ is the all-ones vector.
The rest of \Cref{lemma:increment} uses only the marginal bounds and the
conditioning formula, so its proof applies unchanged.
In particular, taking a high entrywise power requires no additional levels.
Consequently, \Cref{lem:iterate} holds for pseudo-distributions as
well, provided at least two levels remain after all the conditionings.

\begin{proof}[Proof of \Cref{thm:sparse}]
If $N<2^{3(t+3)^2(h+2)}$, a singleton satisfies the required bounds.
Otherwise, compute a level-$r$ pseudo-distribution with $r=2(t+2)+2$ satisfying
\Cref{eq:pseudo-size}. Apply the pseudo-distribution version of
\Cref{lem:iterate} with $t+1$ in place of $t$.
There are at most $2(t+2)$ conditionings, so two levels remain throughout.
The enumeration described above, using $\widetilde\mu$ in place of $\mu$,
therefore finds a set $U$ with
\[
  |U|\ge2^{-3(t+3)^2(h+2)}N=2^{-O(t^2h)}N,
  \qquad d(U)\le2^{-h}.
\]
Both the SDP computation and the enumeration take time $N^{O(t)}$.
\end{proof}

\section{AI methodology}
\label{sec:ai}

This project began with an idea I had been considering for some time:
additive combinatorics can be used to show that a large independent set in
a Cayley graph implies the existence of a large sparse structured set,
which can be found in quasi-polynomial time. More recently, I realized that
the spreadness techniques we developed in \cite{kelley2024explicit} could extend this
approach to general graphs, but I did not know how to make the argument
algorithmic. I used the web version of ChatGPT Astra to help develop the
proof. The interaction proceeded as follows:
\begin{enumerate}
    \item I first asked it to prove the conjecture for Cayley graphs over
    $\F_2^n$, where I already knew a proof, to see whether it could find one
    in the literature or derive one itself. It did neither.
    \item I explained my proof for Cayley graphs and asked it to extend
    the argument to general graphs, without further hints. It did not
    find an extension.
    \item I then suggested using spreadness and sifting
    from \cite{kelley2024explicit} to obtain a density increment, and asked it to work
    out the details. This and further interactions led to the combinatorial results in this paper.
    \item Finally, I asked how to turn these combinatorial observations into an efficient
    algorithm. It identified the connection to SDP relaxations and
    pseudo-marginals, which I had not been aware of.
\end{enumerate}

AI was also used to copy-edit the paper, calculate parameters and do a literature search.

\bibliographystyle{alpha}
\bibliography{indset}

\end{document}